\documentclass[journal]{IEEEtran}
\usepackage{amsmath,amsthm,amssymb,paralist,subfigure,cite,graphicx,color}
\usepackage{algorithm2e}
\newtheorem{lem}{Lemma}
\newtheorem{theorem}{Theorem}
\newtheorem{defn}{Definition}
\newtheorem{rem}{Remark}

\def\mc{\mathcal}

\begin{document}
\title{On the Observability and Controllability of Large-Scale IoT Networks: Reducing Number of Unmatched Nodes via Link Addition}

\author{Mohammadreza Doostmohammadian$^\ast$, Hamid R. Rabiee$^\dagger$, \textit{Senior Member, IEEE}
	
\thanks{		
     	$^\ast$ Faculty of Mechanical Engineering, Semnan University, Semnan, Iran {\texttt{doost@semnan.ac.ir}}.
     	
		$^\dagger$ ICT Innovation Center for Advanced Information and Communication Technology, Computer Engineering Department, Sharif University of Technology, Tehran, Iran {\texttt{rabiee@sharif.edu}}.}}

\maketitle

\begin{abstract}
    In this paper, we study large-scale networks in terms of observability and controllability. In particular, we compare the number of unmatched nodes in two main types of Scale-Free (SF) networks: the Barab{\'a}si-Albert (BA) model and the Holme-Kim (HK) model. Comparing the two models based on theory and simulation, we discuss the possible relation between clustering coefficient and the number of unmatched nodes. In this direction, we propose a new algorithm to reduce the number of unmatched nodes via link addition. The results are significant as one can  reduce the number of unmatched nodes and therefore number of embedded sensors/actuators in, for example, an IoT network. This may significantly reduce the cost of controlling devices or monitoring cost in large-scale systems.
	
\end{abstract}

\section{Introduction} \label{sec_intro}
\IEEEPARstart{I}{nternet} of Things (IoT) emerges as a new technology involved in many aspects of modern life
and recognizing the networking
features behind developing IoT networks is a significant issue \cite{samtani2018identifying,ma2015networking}. This is because the growth in the number of devices connected via IoT poses major challenges in estimation/control of IoT networks. Such IoT networks arise in different scenarios, e.g. in platoon of connected vehicles \cite{lu2014connected},  Cyber-Physical-Systems (CPS) \cite{liu2011cyber,isj19,sadreazami2017distributed,tsipn2020}, smart cities \cite{patel2017using}, distributed inference over sensor networks \cite{kar_IoT}, smart grid monitoring \cite{camsap11,catterson2011embedded}, etc. Tuning the structure of such large-scale man-made networks for observability and controllability purposes may significantly reduce the sensor and actuator placement costs \cite{spl18}. Therefore, analyzing different network structures in terms of observability and controllability may enable us to possibly make  alterations for cost-reduction purposes. In this direction, this paper focuses on network structures which require less number of driver nodes-- the nodes to be injected by a control input-- and observer nodes --the nodes whose states are measured by a sensor.

\textit{Related literature:}
Controllability/observability under structural perturbation is considered in recent literature \cite{insertion17,zhang2019minimal,chen2018minimal,wang2012optimizing,fardad2014optimal,spl17,alcaraz2013structural,mengiste2015effect,van2011decreasing,tnse18}. Authors in \cite{insertion17,zhang2019minimal} consider the problem of optimal cost link addition (deletion) to alter the network to achieve (lose) structural controllability. They claim that this problem is NP-hard and therefore provide approximate solutions. Similarly, the authors in \cite{chen2018minimal} find the minimal number of link additions to render network controllability. The work \cite{wang2012optimizing} presents a graph perturbation method by adding
a minimum number of links to optimize the controllability of  networks such that it can be fully controlled via a single input signal. The authors in \cite{fardad2014optimal} consider the network controllability degradation subject to malicious attacks in the form of link removals. This work finds  links whose removal causes maximal
rank degradation in the  controllability Gramian. Recovering structural observability in case of node failure is considered in \cite{spl17}. In \cite{alcaraz2013structural} authors discuss the problem of maintaining controllability of networks in case of node/link failure by restoring the so-called \textit{Power Dominating Set (PDS)}.
 Authors in \cite{mengiste2015effect} investigate how structural perturbations  affect the network
controllability. They apply a link pruning strategy over different types of real-world and synthetic networks including Scale-Free (SF) and Small-World (SW) random networks. In \cite{van2011decreasing}, it is shown that decreasing the spectral radius of the network by link removal is NP-complete, and therefore, a heuristic approach based on centrality measure is proposed to solve the problem. What is missing in these literature is a comparative study to check the controllability/observability properties of different network models and a study on the relation between particular structural properties of the network and its controllability/observability.

\textit{Motivation and Contribution:} This work is motivated by tuning the number of driver/observer nodes to control/estimate the state of large-scale synthetic networks. This may help to significantly reduce the control/estimation costs of large-scale IoT networks. In this direction, first, we provide a comparative study of two main types of SF networks, namely,  Barab{\'a}si-Albert (BA) model \cite{barabasi_albert1999}  and Holme-Kim (HK) model \cite{Holme2002csf}. These models are known to have similar power-law degree distributions and average shortest paths, while their only difference  stems from different \textit{clustering} properties. We find the relation between  observability/controllability properties and network clustering. This result motivates significant application in designing man-made large-scale IoT networks as one can tune the network clustering (for example by adopting the results of \cite{kashyap2017mechanisms,serrano2005tuning,dehghani2015using}) to reduce the number of required driver/observer nodes and thus improve network controllability/observability. Further, we provide an algorithm to reduce the number of driver/observer nodes in the network by link addition. Assuming the cost of link addition is negligible, this algorithm can be used to reduce controlling and monitoring costs in IoT network applications.
In some literature, as in \cite{pasqualetti2014controllability,cortesi2014submodularity}, the difficulty of the controllability as the required control energy is also considered. In this paper, we focus on the possibility of controlling or estimating the large-scale networks irrespective of the required energy to do so.

\section{The Framework} \label{se_prob}
\subsection{Preliminary Graph-Theoretic Notions} \label{sec_graph}
The controllability and observability of networks are based on graph-theoretic notions introduced in \cite{Liu_nature,asilomar11,tnse18,isj19minimal} which are originally based on the structural approach in \cite{lin}. A network is represented by a  graph $\mc{G}=(\mc{V},\mc{E})$ with $\mc{V}=\{1,...,n\}$ as the set of nodes and $\mc{E}=\{(i,j)|i \leftrightarrow j\}$ as the set of links. Note that the links are considered bidirectional as the networks in this paper are undirected. Denote the degree of node $i$  by $d_i$, and  the set of neighbors of node $i$ by $\mc{N}_i$. Let $\Gamma = (\mc{V}^+,\mc{V}^-,\mc{E}_\Gamma)$ denote a bipartite representation of graph $\mc{G}$ with set of nodes $\mc{V}^+=\mc{V}^-=\mc{V}$ and the set of links $\mc{E}_\Gamma=\{(i^+,j^-)|(i,j) \in \mc{E}\}$. Note that all the links in $\mc{E}_\Gamma$ are directed from  $\mc{V}^+$ to $\mc{V}^-$. Define a \textit{maximum matching}, $\mc{M}$, as the maximal set of links with no shared end-node and begin-node. In the bipartite graph $\Gamma$, the maximum matching is defined as the maximum number of links from $\mc{V}^+$ to $\mc{V}^-$ sharing no begin-node in $\mc{V}^+$ and no end-node in $\mc{V}^-$. Note that the the maximum matching $\mc{M}$ is not unique in general. Refer to \cite{murota} for more information on these graph-theory notions.

\begin{defn}
	Define the \textit{matched nodes}, $\rho\mc{M}$, as the nodes incident to $\mc{M}$. Define the \textit{unmatched nodes} as $\delta\mc{M} = \mc{V}\backslash\rho\mc{M}$. The set of unmatched nodes is not unique, which is a consequence of the fact that matching is not unique. The set of \textit{contractions} (\textit{dilations}) characterize the set of all equivalent unmatched nodes in terms of observability (controllability).
\end{defn}



A driver node is defined as the node in the network that must be injected by proper control input for network controllability \cite{Liu_nature,isj19minimal}. Similarly, an observer node is defined as the node whose state must be measured for observability and estimation \cite{tnse18}. Since the network is undirected, the dual concepts of controllability and observability are equivalent and the number of observer nodes and driver nodes  are equal. It is known that the number of driver nodes (observer nodes) may be considered as a measure of controllability  (observability) \cite{liu2012controlcentrality}. \textit{It is proved that for network controllability (observability) the unmatched nodes work as driver nodes (observer nodes) \cite{Liu_nature,tnse18,isj19minimal}}. In other words, by injecting control input to the unmatched nodes one may control the network, and by measuring the state of the unmatched nodes the state of the entire network is observable to the estimator. In fact, the number of unmatched nodes determine the minimal number of inputs/outputs for controllability/observability. In this direction, the number of unmatched nodes represent a controllability/observability metric for complex networks \cite{liu2012controlcentrality}.

\subsection{Problem Statement}
In this work, two main types of Scale-Free (SF) networks are analyzed in terms of number of unmatched nodes: (i) Barab{\'a}si-Albert (BA) model, and (ii) Holme-Kim (HK) model. These two types of networks are similar in terms of power-law degree distribution prevalent in many real-world networks. However, the HK model represents a more clustered structure as compared to the BA model. Therefore, the higher clustering coefficient may affect the controllability and observability of the network while other network parameters such as  average node degrees are the same. The result is significant as in man-made IoT network one may reduce the number of driver nodes for controllability or observer nodes for estimation by tuning the clustering of the large-scale IoT network and, therefore, reduce the related costs. In this direction, this paper provides a new algorithm to reduce the number of unmatched nodes by suitable link addition in the network. As a case study, we show that for an example IoT network the proposed algorithm significantly reduces the number of unmatched nodes.

\section{Network Models} \label{sec_net}
\subsection{Barab{\'a}si-Albert Model}
The BA model \cite{barabasi_albert1999} refers to an algorithm for generating random SF networks. The SF network represents a network having a \textit{power-law} degree distribution, i.e. for the fraction $P(d)$ of nodes,
   $P(d) = d^{-\alpha}$,
where $d$ is the node degree  \cite{barabasi2003scale}. 
The main feature of SF networks is their few number of high-degree nodes (or hubs) connected to smaller hubs and iteratively followed by very low-degree nodes.

Barab{\'a}si and Albert proposed a generative random mechanism to model such  power-law degree distributions, which is based on the so-called \textit{preferential-attachment} method as follows. In every iteration a node is added to an initial \textit{seed graph} of few nodes. This new node  makes $L$ new connections to the old nodes in the graph. The link connections are random and the probability of connecting to an old node $i$ is proportional to its node-degree $d_i$. In other words, the new node preferably connects to high-degree nodes, thus, implying the name preferential attachment. Based on this model, the few hubs get more and more connectivity, while many low-degree nodes remain in the network. Algorithm~\ref{alg_BA} briefly states the mechanism of BA network model.
\begin{algorithm} \label{alg_BA}
	\textbf{Given:} seed graph $\mc{G}$ of size $m<n$, final size of the graph $n$, number of added links $L$ in each iteration
	
	\For{$k=m+1:n$ }{
		$\bullet$ Find degree of every node $i$ denoted by $d_i$\;
		$\bullet$ Assign probability $p_i$ to every node $i$ as $p_i = d_i/\sum_{j =1}^{k-1}d_j$\;
		$\bullet$ Add a node $a$ to graph $\mc{G}$\;
		\For{$l=1:L$ }{
			$\bullet$ Randomly choose a node $b$ based on the set of assigned probabilities $p_i$\;
			$\bullet$ Add a link between nodes $a$ and $b$\;}
	}
	\textbf{Return} BA graph $\mc{G}$ of size $n$\;\
	
	\caption{Pseudo-code for Barab{\'a}si-Albert (BA) network generation from a seed graph. }
\end{algorithm}
Assuming $m \ll n$, the average node degree in Algorithm~\ref{alg_BA} is approximately equal to $2L$. Let define a network characteristic known as \textit{clustering coefficient} or \textit{$2$nd-order clustering} as the number of closed triplets  (also referred to as \textit{triangle formation}, \textit{triangular motif}, or \textit{$3$-clique} \cite{benson2016higher}) to the total number of connected open triplets (also known as \textit{$2$-wedges}) in the network.  It is known that for BA networks the clustering coefficient is independent of the node degree $d$ as follows \cite{szabo2003structural}:
\begin{equation}\label{eq_C_BA}
C(d) = \frac{L-1}{8}\frac{(\log(n))^2}{n}.
\end{equation}


\subsection{Holme-Kim Model}
Holme and Kim propose another mechanism for generating SF networks known as HK model \cite{Holme2002csf}, which further developed in \cite{Toivonen2006social}. This mechanism is also based on the preferential attachment of BA model with some modifications. This mechanism also starts with a seed graph. In this model, iteratively a node is added to the graph making $L_1$ random connections to the old nodes based on preferential attachment. However, in the HK model, the new node also makes $L_2$ connections to the neighbors of the preferentially attached node in the last step. These new $L_2$ random connections are also made based on the degree of the neighboring nodes. The new step of preferential attachment to the neighboring nodes generates \textit{clusters} in the network which are missing in the BA model. Therefore, the adopted approach in the HK model is also known as \textit{clustered preferential attachment} method. The procedure for HK network model is summarized in Algorithm~\ref{alg_HK}. 
Assuming $m \ll n$, the average node degree in Algorithm~\ref{alg_HK} is approximately equal to $2L=2(L_1+L_1\times L_2)$. The clustering coefficient of the HK network is as follows \cite{szabo2003structural}:
	\begin{equation} \label{eq_C_HK}
	C(d) \approx  \frac{4L_2}{d}+\frac{L-1}{8}\frac{(\log(n))^2}{n},
	\end{equation}	
where $d$ is the node degree. By comparing this equation with equation \eqref{eq_C_BA} it is clear that the clustering coefficient is increased for $L_2\geq 1$. This implies that by increasing the triad formations via the parameter $L_2$ one can increase the clustering coefficient of the HK network.	

\begin{algorithm} \label{alg_HK}
	\textbf{Given:} seed graph $\mc{G}$ of size $m<n$, final size of the graph $n$, number of added links $L_1$ and $L_2$  in each iteration
	
	\For{$k=m+1:n$ }{
		$\bullet$ Find degree of every node $i$ denoted by $d_i$\;
		$\bullet$ Assign probability $p_i$ to every node $i$ as $p_i = d_i/\sum_{j =1}^{k-1}d_j$\;
		$\bullet$ Add a node $a$ to graph $\mc{G}$\;
		\For{$l=1:L_1$ }{
			$\bullet$ Randomly choose a node $b$ based on the set of assigned probabilities $p_i$\;
			$\bullet$ Add a link between nodes $a$ and $b$\;
			$\bullet$ Find $\mc{N}_b$ as the neighbors of node $b$\;
			$\bullet$ Assign a probability $q_i$ to every neighbor node $i$ as $q_i = d_i/\sum_{j \in \mc{N}_b}d_j$\;
	     	\For{$h=1:L_2$ }{
		    	$\bullet$ Randomly choose a node $c$ from $\mc{N}_b$ based on the set of assigned probabilities $q_i$\;
		    	$\bullet$ Add a link between nodes $a$ and $c$\;}			
			}
	}
	\textbf{Return} HK graph $\mc{G}$ of size $n$\;\
	
	\caption{Pseudo-code for Holme-Kim (HK) network generation from a seed graph. }
\end{algorithm}

\begin{rem} \label{rem_net}
	Based on Algorithm~\ref{alg_BA} and \ref{alg_HK}, the main difference of BA and HK models is in the second \textbf{for} loop in Algorithm~\ref{alg_HK}, which results in higher number of \textit{closed triplets} in the HK model. Note that this holds for the same number of added links $L=L_1+L_1\times L_2$ and, therefore, similar average node degree for BA and HK networks of the same size.
\end{rem}

Following Remark~\ref{rem_net}, the mentioned difference of BA and HK models results in difference of clustering coefficient.
Following equations \eqref{eq_C_BA} and \eqref{eq_C_HK}, the clustering coefficient of the HK network is greater than the BA network of the same size (same number of nodes and links) \cite{Holme2002csf,szabo2003structural}.

\section{Comparing Number of Observer/Driver Nodes} \label{sec_main}

In this section, we compare the number of unmatched nodes in the two main models for SF networks introduced in the previous section. In this direction, the main network parameters namely power-law degree distribution
and logarithmically increasing average shortest-path length along with number of nodes, number of links, and average node-degrees are similar in both BA and HK models. The only difference stems from the clustering coefficient. Recall that the clustering coefficient is greater in the HK model as compared to the BA model.

\subsection{Theoretical Results}
This subsection provides theoretical results and comparison on the matching properties of the HK and BA models.
\begin{theorem} \label{thm_hall_ore} (Hall-Ore Theorem \cite{murota})
	The size of maximum matching $\mc{M}$ in the bipartite graph $\Gamma$ (the bipartite representation of graph $\mc{G}$) is as follows:
	\begin{eqnarray}
	|\mc{M}|=\min\{|\Lambda(\bar{X})|-|\bar{X}|,\bar{X} \subseteq \mc{V}^{+}\}+|\mc{V}^{+}|,
	\end{eqnarray}
	where $\Lambda(\bar{X})=\{v \in \mc{V}^{-}|\exists u \in \bar{X}: (u^+,v^-)\in \mc{E}_\Gamma \}$.
\end{theorem}
In fact,  $\Lambda(\bar{X})$ denotes the set of nodes in $\mc{V}^{-}$ adjacent to nodes in $\bar{X} \subseteq \mc{V}^{+}$. From Theorem~\ref{thm_hall_ore} we have,
\begin{eqnarray}
|\delta\mc{M}|=|\mc{V}^{+}|-|\mc{M}|=\max\{|\bar{X}|-|\Lambda(\bar{X})|,\bar{X} \subseteq \mc{V}^{+}\}.
\end{eqnarray}
\begin{lem} \label{lem_triplet}
	Let denote the bipartite representation of the closed triplet $1 \leftrightarrow 2 \leftrightarrow 3 \leftrightarrow 1$ by ${\Gamma}_1$ and the open triplet $1 \leftrightarrow 2 \leftrightarrow 3$ by ${\Gamma}_2$. Then, $|\delta\mc{M}_2|>|\delta\mc{M}_1|$, where $\delta\mc{M}_1$ and $\delta\mc{M}_2$ respectively represent the number of unmatched nodes in $\Gamma_1$ and $\Gamma_2$.
\end{lem}

\begin{proof}
    In $\Gamma_1$ for $|\bar{X}|=1$ we have $|\Lambda(\bar{X})|=2$, for $|\bar{X}|=2$ we have $|\Lambda(\bar{X})|=3$, and for $|\bar{X}|=3$ we have $|\Lambda(\bar{X})|=3$. Therefore, $|\delta\mc{M}_1|=0$. In $\Gamma_2$ for $\bar{X}=\{1,3\}$ we have $\Lambda(\bar{X})=\{2\}$ and $|\Lambda(\bar{X})|=1$. For other choices of $\bar{X}$ we have $|\bar{X}|-|\Lambda(\bar{X})|<1$. Therefore, $|\delta\mc{M}_2|=1>|\delta\mc{M}_1|$.
\end{proof}
Lemma~\ref{lem_triplet} implies that more open triplet in the network results in higher number of unmatched nodes. Therefore, one expect to see less unmatched nodes in HK networks as compared to BA networks. This is because in the HK model, as compared to the BA model, some open triplets are replaced with closed triplets.

\subsection{Monte-carlo Simulation}
To compare the number of unmatched nodes a Monte-Carlo simulation is performed in this section. We generate the BA model with parameter $L=2$ links per iteration and the HK model with $L_1=L_2=1$ link addition per iteration, and adding one node per iteration for both models. Since $L=L_1+L_1\times L_2$ the number of links and average node degrees are the same in both models. The simulations are performed over $100$ realizations of BA and HK models from $100$ nodes to $1200$ nodes and the results are shown in Fig.~\ref{fig_unmatched}.

 \begin{figure} [hbpt!]
 	\centering
 	\includegraphics[width=1.7in]{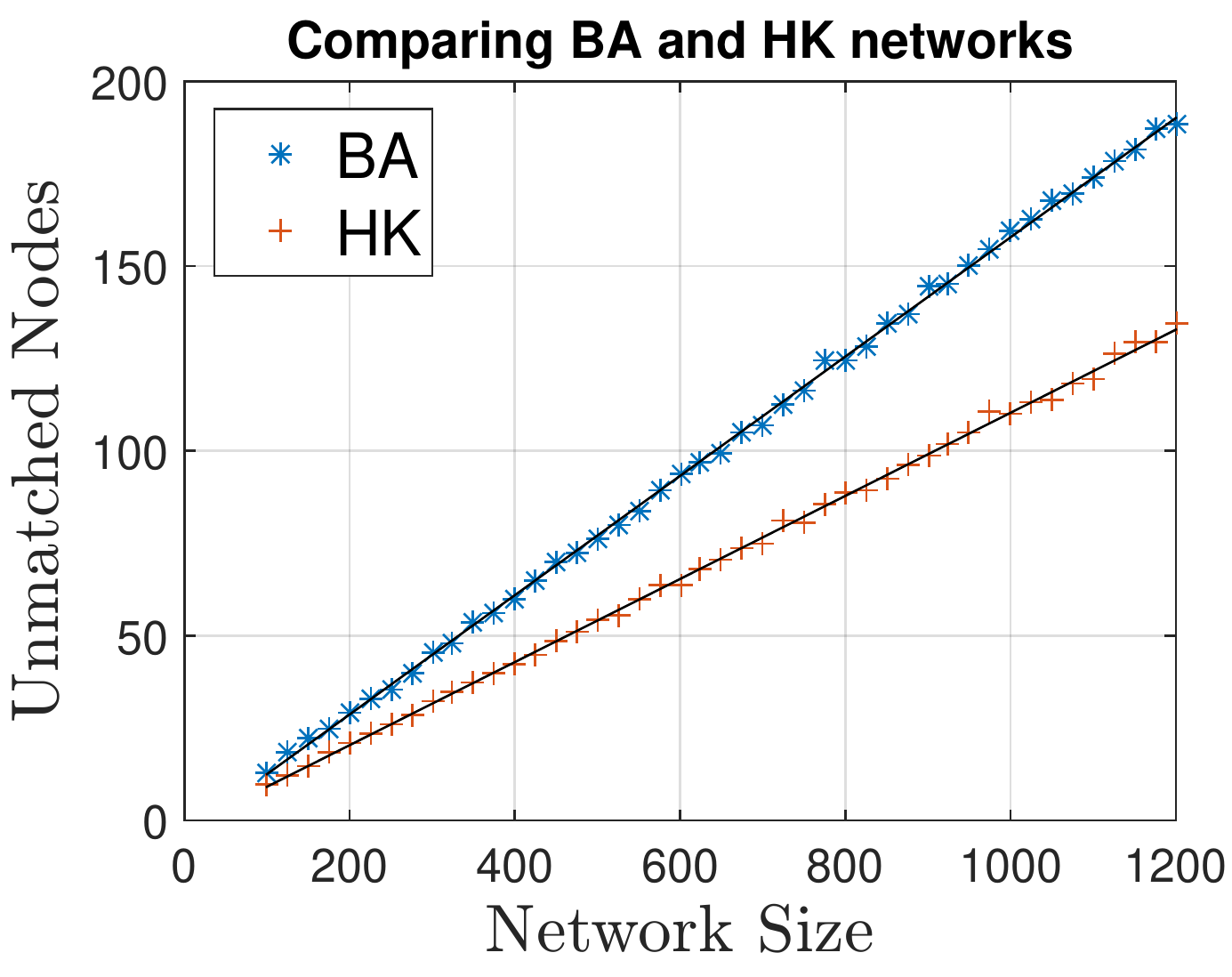}
 	\includegraphics[width=1.7in]{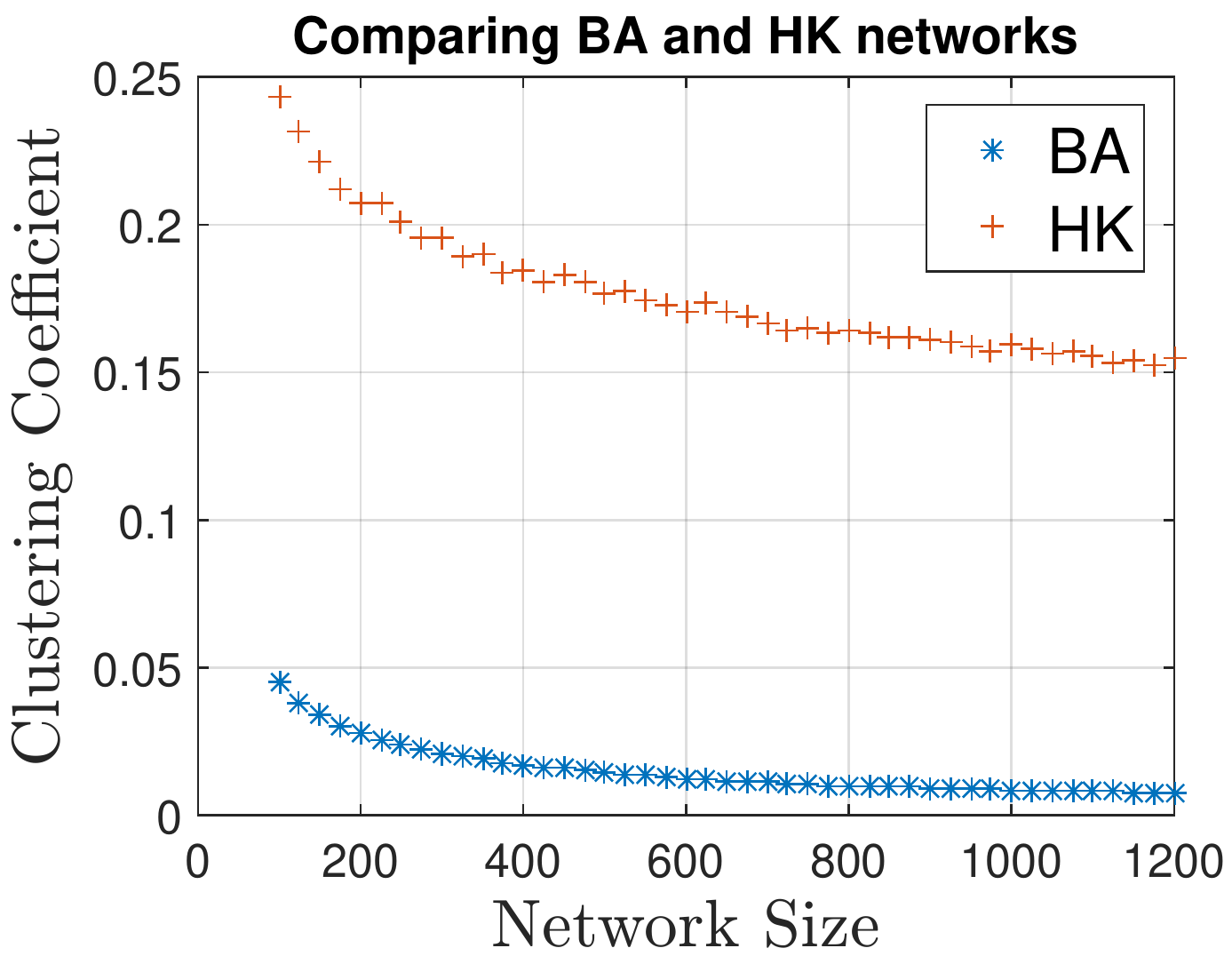} 	
 	\caption{(Left) This figure shows the number of unmatched nodes for different network sizes of BA and HK models. As it can be seen the number of unmatched nodes in the BA model is more than the HK model of the same size. (Right) This figure compares the normalized clustering coefficient of BA and HK network models. It is clear that the clustering coefficient of the HK network is greater than the BA network of the same size.}
 	\label{fig_unmatched}
 \end{figure}

 As expected, from Fig.~\ref{fig_unmatched}-(Right), the network clustering of HK model is higher than the network clustering of the BA model (with the same number of nodes). From Fig.~\ref{fig_unmatched} it can be seen that  by increase in the network size, the clustering is decreased while the number of unmatched nodes is increased for both models. Further, the number of unmatched nodes in the HK networks are less than the BA networks of the same size. Having that the clustering coefficient of the BA network is less than the HK network implies an inverse correlation between the network clustering and number of unmatched nodes. In other words,  increase in the clustering coefficient results in decrease in the number of unmatched nodes and vice versa.
 It is important to note that the average node degree plays a key role in the number of unmatched nodes \cite{Liu_nature}; therefore, for comparison of the two network types, the average node degree for both networks is the same. We again emphasize that the HK and BA models of SF networks are similar in terms of most graph features except the clustering.

\section{Reducing Unmatched Nodes by Link Addition} \label{sec_reduce}
In this section, we provide an algorithm to reduce the number of unmatched nodes in the network via link addition. The algorithm is based on the results of the previous section on triangular motif formation, i.e. increasing the number of closed triplets by link addition. This algorithm targets specific network components related to the connected open triplets or 2-wedges. These components, previously introduced in Section~\ref{sec_graph}, are known as contractions/dilations. In this section, we provide our algorithm based on contractions and the result can be easily extended to the case of dilations. Define a contraction ${\mc{C}_i}$ as the set of  nodes in the network for which $|\mc{N}_{\mc{C}_i}|<|\mc{C}_i|$, where $|\mc{N}_{\mc{C}_i}|$ denotes the cardinality of the set of neighbors of $\mc{C}_i$. The Dulmage-Mendelsohn (DM) decomposition algorithm \cite{murota} can be used to find the contractions in the network by computational complexity $\mc{O}(n^{2.5})$. Based on the definition, it is clear that two successive nodes in a contraction make an open triplet with their common neighboring node. Therefore, adding a link between the successive nodes in a contraction makes a closed triplet (or a triangular motif) and increases the clustering coefficient. This in turn reduces the number of unmatched nodes in the network. In this direction, the main algorithm of this section is given as  Algorithm~\ref{alg_reduce}.

\begin{algorithm} \label{alg_reduce}
	\textbf{Given:} graph $\mc{G}=(\mc{V},\mc{E})$, contractions $\mc{C}_i,~i=\{1,...,n_{\mc{C}}\}$, number of link addition $T<\sum_{i=1}^{n_{\mc{C}}} |\mc{C}_i|/2$
	
	$\bullet$ reorder the contractions from smaller size to larger size\;
	\For{$k=1:T$ }{
		$\bullet$ Find two successive nodes $i$ and $j$ in $\mc{C}_k$\;
		$\bullet$ Make  a bidirectional link between $i$ and $j$ \;		
	}
	\textbf{Return} Graph $\mc{G}_2=(\mc{V},\mc{E}_2)$ with  $|\mc{E}_2| = |\mc{E}| + T$\;\
	
	\caption{Pseudo-code for the proposed link addition algorithm to reduce the number of unmatched nodes in the network. Note that $n_{\mc{C}}$ in the first line represents the number of contractions. Further, by two successive nodes in the contraction we imply two nodes in the contraction that share a neighboring node.}
\end{algorithm}
 Note that in Algorithm~\ref{alg_reduce} the first line reorders the contractions based on their sizes. The reason is that for smaller contractions  less number of links are needed to connect the nodes for triangle formation. This result in more reduction of unmatched nodes via less number of link addition. Further, since some contractions may include more than two nodes, we need to add a link between every two distinct successive nodes in the contraction. Therefore, for a contraction $i$ of size $|\mc{C}_i|$, we need to add $|\mc{C}_i|/2$ links and in general, for all $n_{\mc{C}}$ contractions in the network, the maximum possible link additions is equal to $\sum_{i=1}^{n_{\mc{C}}} |\mc{C}_i|/2$. Also, recall that the complexity of Algorithm~\ref{alg_reduce} mainly stems from the complexity of DM decomposition (of order $\mc{O}(n^{2.5})$). Therefore the polynomial order complexity ensures scalability over large-scale network applications.
 
\begin{theorem} \label{thm_algorithm}
	For a given network $\mc{G}=(\mc{V},\mc{E})$, adding  set of links $\mc{E}_{T}$ (including $T$ links) based on the Algorithm~\ref{alg_reduce} reduces the number of unmatched nodes by at least $T$ nodes in $\mc{G}_2=(\mc{V},\mc{E}_2)$ with $\mc{E}_2 = \mc{E} \cup \mc{E}_{T}$.
\end{theorem}
\begin{proof}
	Let $\delta \mc{M}_1$ and $\delta \mc{M}_2$ respectively represent the unmatched nodes in $\mc{G}$ and $\mc{G}_2$.
	Based on the definition of a contraction $\mc{C}_i \subset \mc{G}$ for the neighbors of this contraction $\mc{N}_{\mc{C}_i}$ we have $|\mc{N}_{\mc{C}_i}|<|\mc{C}_i|$. Therefore, there exist two successive nodes $u,v \in \mc{C}_i$ sharing one common neighbor $w \in \mc{N}_{\mc{C}_i}$ such that $u \leftrightarrow w \leftrightarrow v$ makes an open triplet. Following Lemma~\ref{lem_triplet}, this open triplet includes $1$ unmatched node. Adding an undirected link $(u,v)\in \mc{E}_{T}$ according to Algorithm~\ref{alg_reduce} makes a closed triplet $u \leftrightarrow w \leftrightarrow v \leftrightarrow u$ which, according to Lemma~\ref{lem_triplet}, includes no unmatched node. Therefore, adding every such link according to Algorithm~\ref{alg_reduce} reduces $|\delta \mc{M}_1|$ by at least $1$ and consequently by adding $T$ such links we have $|\delta \mc{M}_2| \leq |\delta \mc{M}_1| - T$ which proves the theorem.
\end{proof}
Recall that following the proof of Theorem~\ref{thm_algorithm}, increasing the number of closed triplets via Algorithm~\ref{alg_reduce} directly results in the increase of parameter $L_2$ in equation~\eqref{eq_C_HK} and, therefore, increase in the network clustering coefficient.

\section{Simulations} \label{sec_example}
\subsection{Illustrative Example}
Consider two networks of  $25$ nodes, one generated by the BA model with $L=2$ and the other generated by the HK model with $L_1=L_2=1$. The two networks are shown in Fig.~\ref{fig_sim_BA}. For both networks the seed graph is considered as a line graph of $5$ nodes.
\begin{figure}[b]
	\centering
	\includegraphics[width=1.4in]{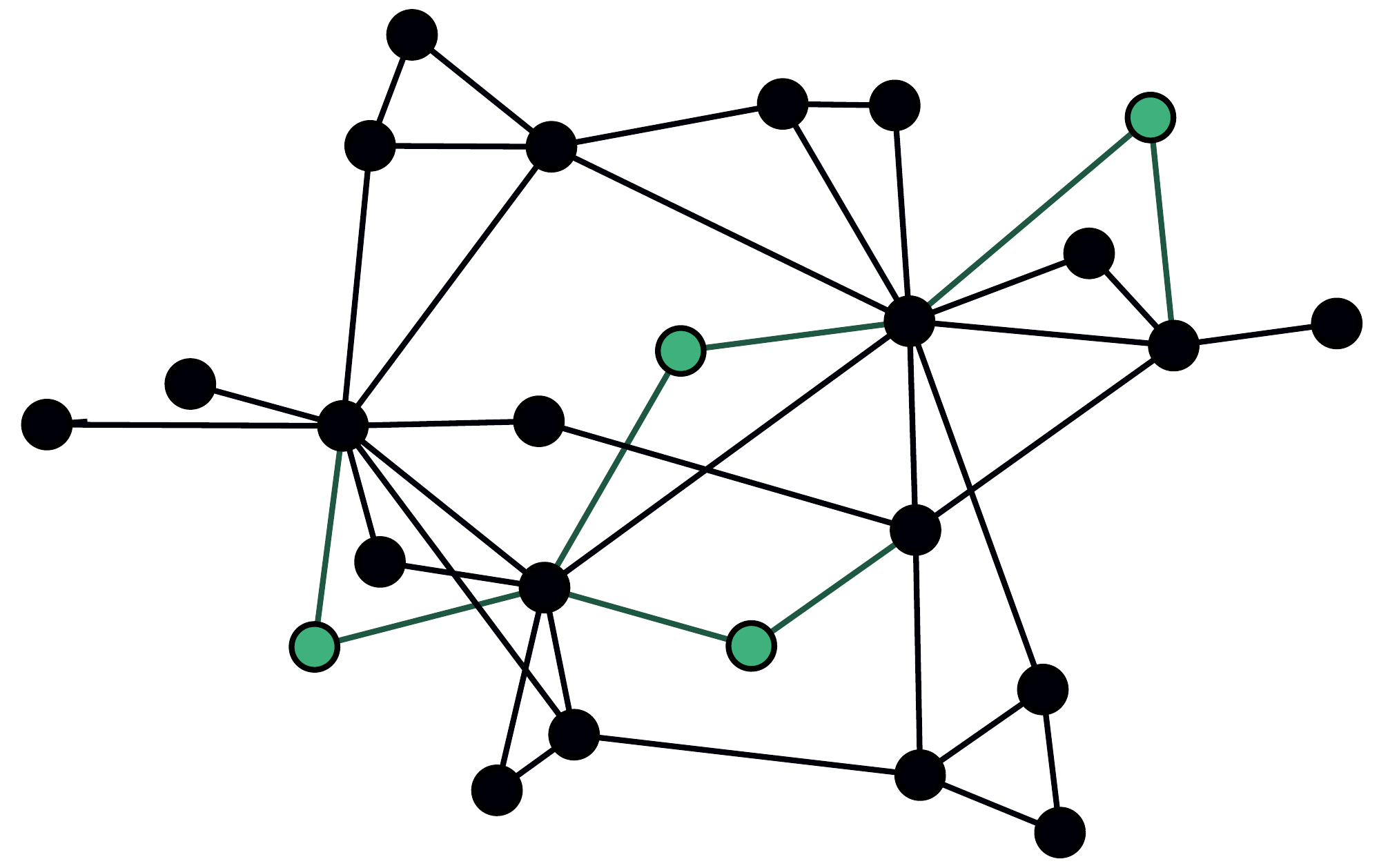}
	\includegraphics[width=1.4in]{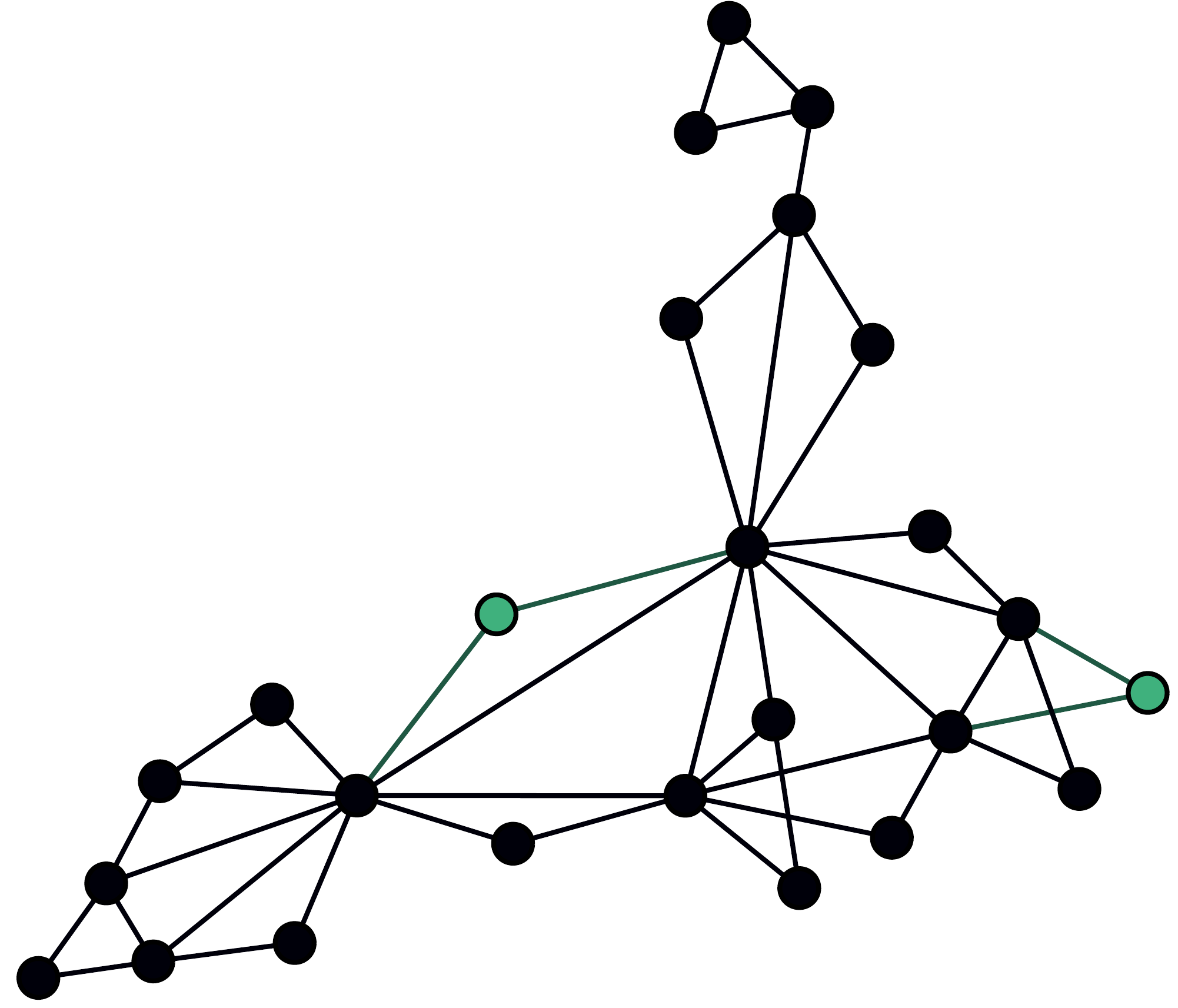}
	\caption{An example (Left) BA network and (Right) HK network of $25$ nodes is shown in this figure. The $4$ unmatched nodes in BA network and $2$ unmatched in HK model are shown as green nodes.}
	\label{fig_sim_BA}
\end{figure}
%
The properties of the two networks are compared in Table~\ref{tab_sim}.  As expected, the HK network with higher clustering coefficient includes less number of unmatched nodes compared to the BA network.
\begin{table} [hbpt!]
	\centering
	\caption{Properties of the BA and HK networks given in Fig.~\ref{fig_sim_BA}.}
	\begin{tabular}{|c|c|c|c|}
		\hline
		Network& Avg. deg.& Unmatcheds & Clust. coef.\\
		\hline
		BA & ~$3.4$  &~$4$ &~$0.228$ \\
		\hline
		HK & ~$3.4$  &~$2$ &~$0.330$ \\
		\hline
		\hline
	\end{tabular}
	\label{tab_sim}
\end{table}

Monte-Carlo simulation of Kalman filtering is performed over the two networks. Consider $x$ to represent the state of each node. The link weights, representing the non-zero entries of system matrix $A$, are chosen randomly. Note that the entries are such that $\rho(A)=1.2$ ($\rho(A)$ represents the spectral radius of system matrix $A$) implying unstable system dynamics $\dot{x}=Ax+\nu$. We measure the states of the colored nodes in Fig.~\ref{fig_sim_BA} as $y=Cx+\eta$ with the non-zero entries of the measurement matrix $C$ chosen randomly. The system noise $\nu$ and measurement noise $\eta$ are considered as Gaussian noise $\mc{N}(0,0.02)$.  The mean-squared-estimation-error (MSEE) for states of both networks are shown in Fig.~\ref{fig_kalman_BA}. As expected,  measurement of unmatched nodes as observer nodes results in structural observability and, therefore, the MSEE in both networks are bounded steady-state stable.

\begin{figure}
	\centering
	\includegraphics[width=1.5in]{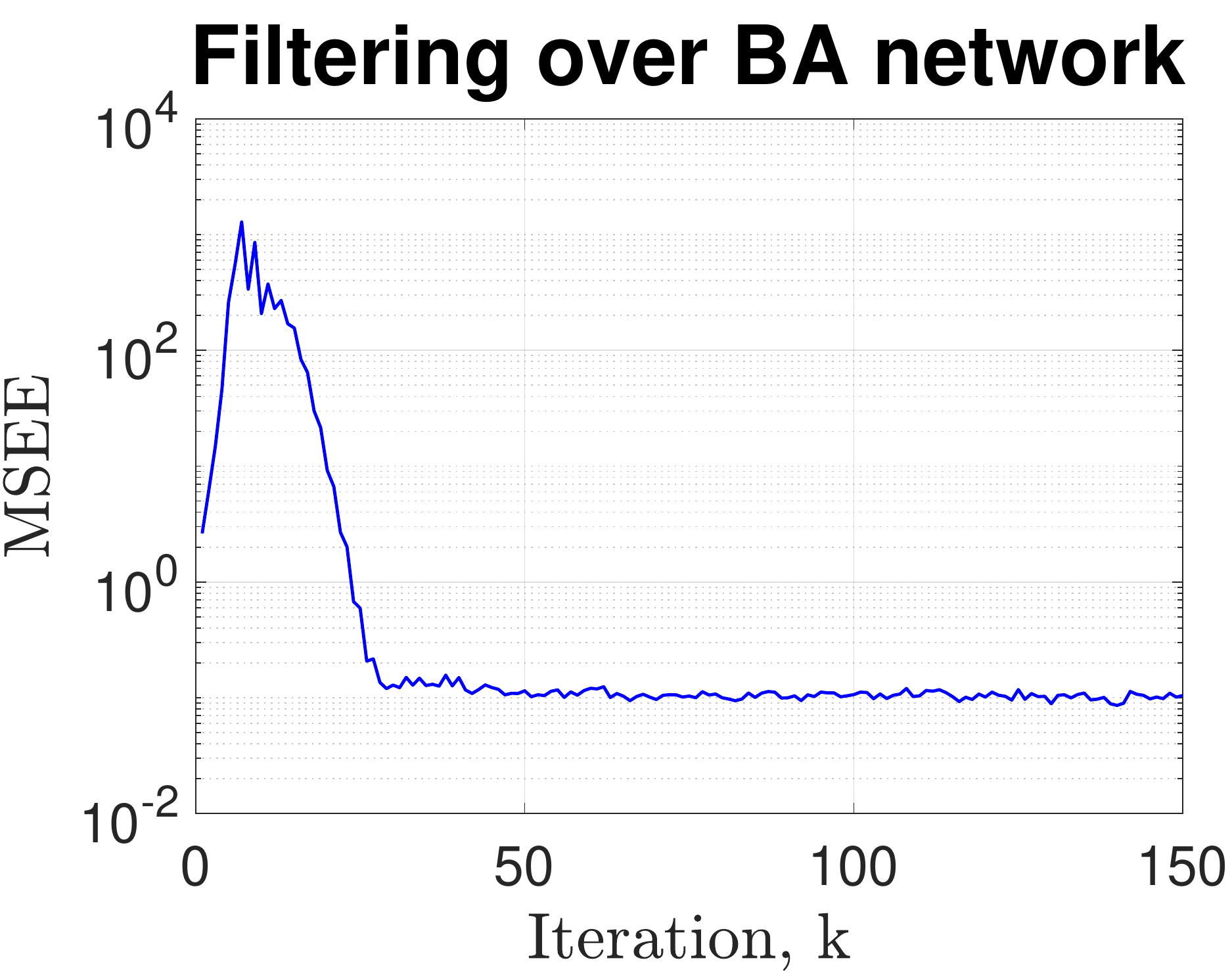}
	\includegraphics[width=1.5in]{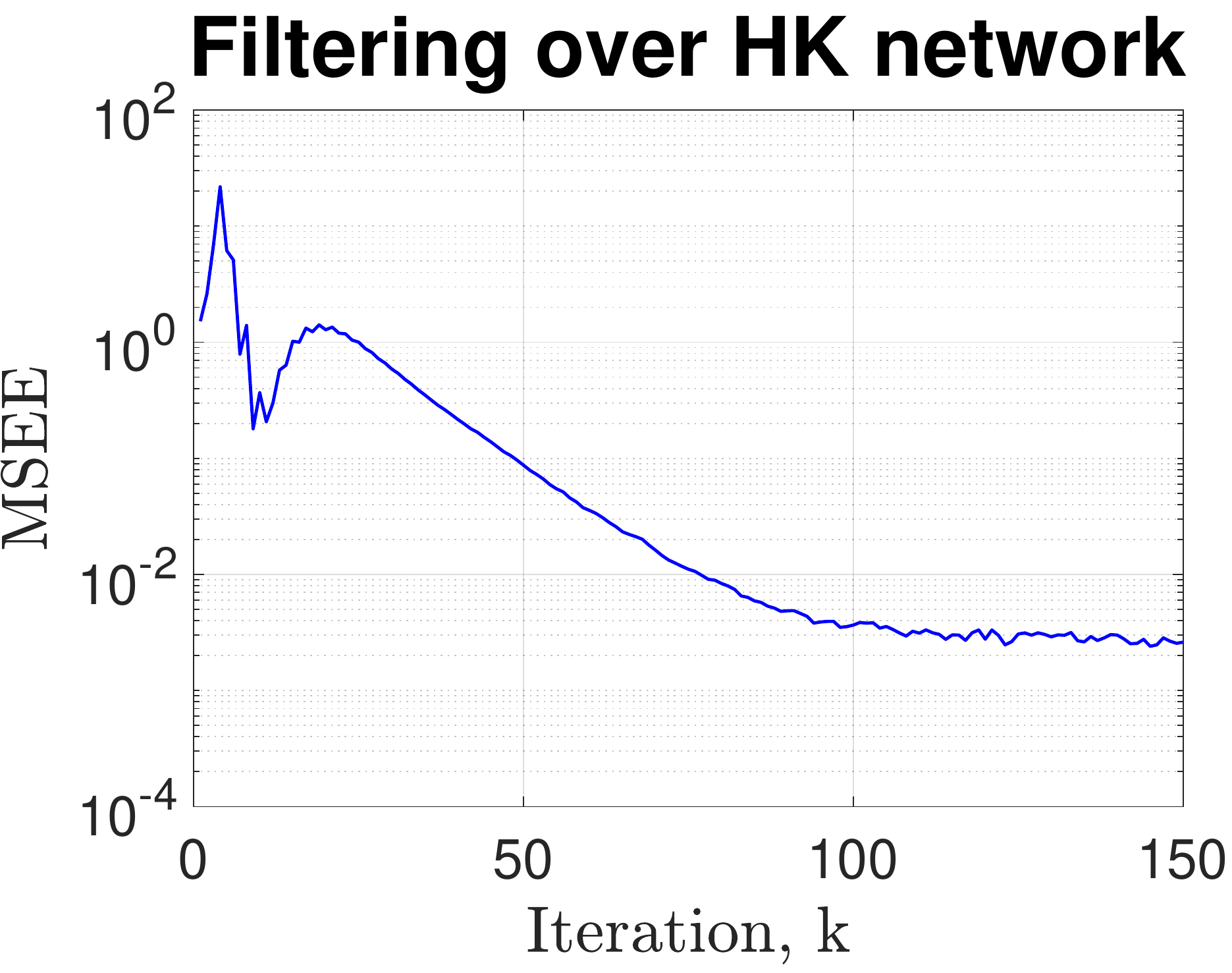}
	\caption{The MSEE of the Kalman filter over the network given (Left) in Fig.~\ref{fig_sim_BA}-(Left) with the $4$ unmatched nodes measured, and (Right) in Fig.~\ref{fig_sim_BA}-(Right) with the $2$ unmatched nodes measured.}
	\label{fig_kalman_BA}
\end{figure}

\subsection{A Case Study}
In this subsection, we apply Algorithm~\ref{alg_reduce} on a real-world network, known as \textit{Route-views}, to reduce its number of unmatched nodes. The $6474$ nodes of this undirected network represent autonomous systems of the Internet connected with each other via $13895$ communication links. This network is shown in Fig.~\ref{fig_rout}, where the red nodes represent the unmatched nodes.
\begin{figure}
	\centering
	\includegraphics[width=3.0in]{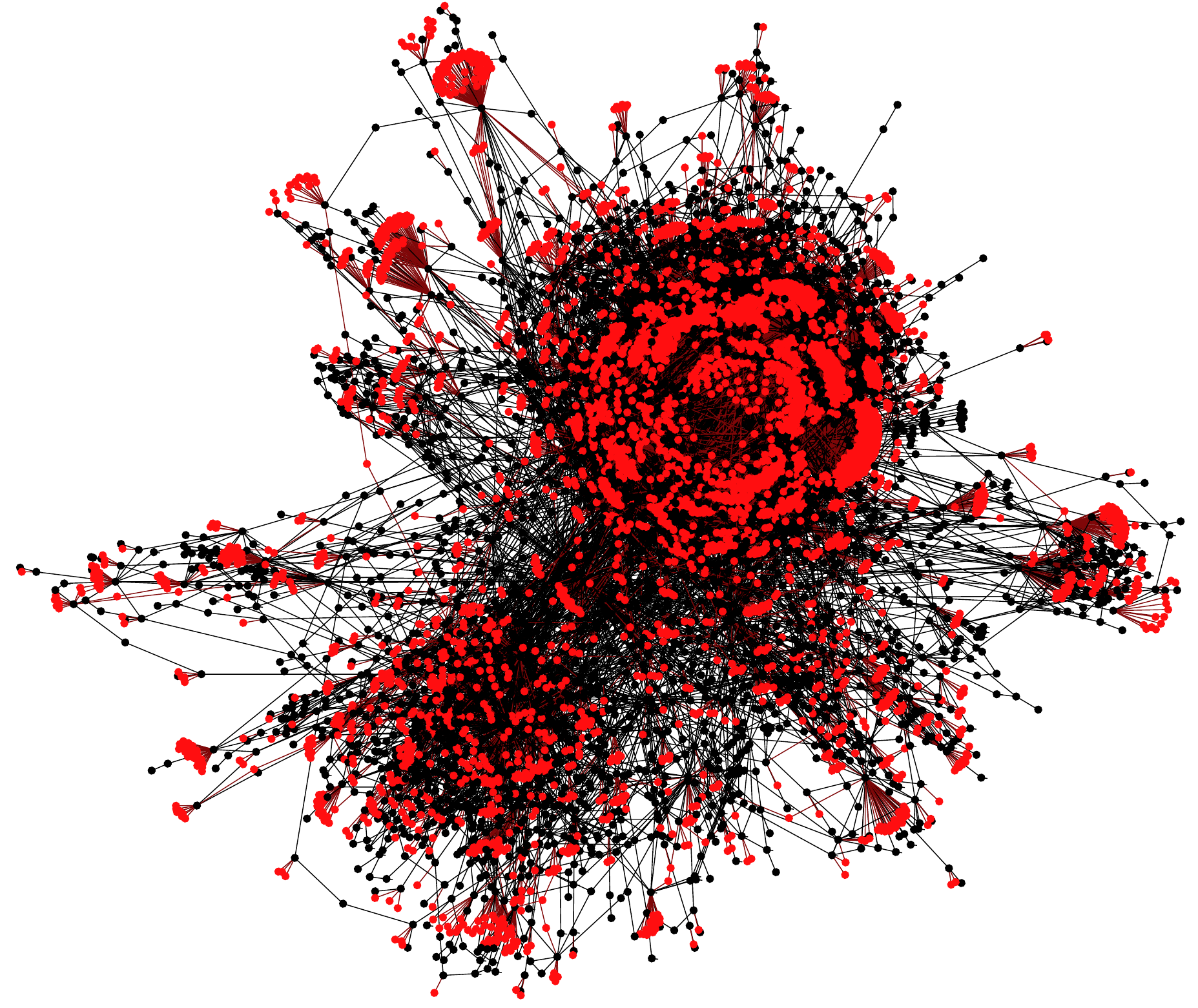}
	\caption{This undirected network represents the communications of autonomous systems, also  known as Route-views network. The red nodes represent one set of possible  unmatched nodes in the network.}
	\label{fig_rout}
\end{figure}
More information on this network is available in \cite{leskovec2007graph}. We add new links to this network using Algorithm~\ref{alg_reduce}, and check the number of unmatched nodes in the network.  The results are shown in Table~\ref{tab_reduce}.
 \begin{table} [hbpt!]
 	\centering
 	\caption{Reduced number of unmatched nodes via link addition by making triangular motifs.}
 	\begin{tabular}{|c|c|c|}
 		\hline
 		Added links& Unmatcheds & Clustering coefficient $\times 10^3$\\
 		\hline
 		$0$ & ~$3568$  &~$9.591$  \\
 		\hline
 		$30$ & ~$3512$  &~$9.635$  \\
 		\hline
 		$60$ & ~$3461$  &~$9.679$ \\
 		\hline
 		$90$ & ~$3413$  &~$9.724$ \\
 		\hline
 		$120$ & ~$3369$  &~$9.768$ \\
 		\hline
 		\hline
 	\end{tabular}
 	\label{tab_reduce}
 \end{table}
As expected, the increase in the number of closed triplets result in increase of the clustering coefficient and in turn decrease in the number of unmatched nodes. This is significant as one can reduce the number of observer nodes for estimation (or the number of driver nodes for control) by suitably adding links to similar IoT networks. Note that in this work we assume the cost of link addition is less than the cost of sensor/actuator placement. This is a key assumption in the relevant literature \cite{insertion17,zhang2019minimal,chen2018minimal,wang2012optimizing}.  Therefore, recalling that adding every link reduces at least one unmatched node (to be controlled/observed by an actuator/sensor), the overall cost is reduced by every link addition.

\section{Conclusions} \label{sec_con}
Results of this paper are significant as in large-scale networks, for example in smart grid monitoring \cite{camsap11}, the number of embedded sensors and related costs can be reduced. For example, following the results in \cite{camsap11}, the design of smart grid network is a key measure affecting the system observability. In this direction, Algorithm~\ref{alg_reduce} and similar results for increasing the number of triad formations can be used in the design of such  smart grid networks  to reduce the number of sensor placement  embedded on the grid for  observability recovery.

As a direction of future research other types of random networks, including Small World (SW) and Erdos-R{\'e}nyi (ER), can be considered for similar analysis. In this regard, the algorithm for constructing these models needs to be modified to tune the clustering and number of closed triplets while keeping the average node degree unchanged.


\bibliographystyle{IEEEbib}
\bibliography{bibliography}
\end{document}